\documentclass[a4paper,10pt]{article}

\usepackage{amsmath,amsfonts,amssymb}
\usepackage{amsthm}
\usepackage{mathtools}
\usepackage{amsopn}
\usepackage{tabularx,lipsum,environ}
\usepackage{paralist}
\usepackage{microtype}
\usepackage{authblk}
\usepackage{hyperref}
\usepackage{fancyhdr}
\usepackage{rotating}
\usepackage{overpic}
\usepackage{ucs}
\usepackage{enumerate}
\usepackage{graphicx}
\usepackage{booktabs}
\usepackage{varwidth}
\usepackage{subfig}
\usepackage{todonotes}
\usepackage{multicol}
\usepackage{multirow}
\usepackage{lineno}
\usepackage{comment}

\usepackage{setspace}
\providecommand{\keywords}[1]
{ \medskip  
  \noindent
  \small	
  \textbf{\textit{Keywords---}} #1
}

\newcommand{\cost}{{\rm{cost}}}
\newcommand{\PoA}{{\rm{PoA}}}
\newcommand{\PoS}{{\rm{PoS}}}
\newcommand{\NE}{{\rm{NE}}}
\newcommand{\s}{\mathbf{s}}
\newcommand{\ord}{\mathrm{ord}}
\newcommand{\revised}[1]{\textcolor{black}{#1}}
\newcommand{\PoSLsum}{n+{1}/{n}-1}

\newtheorem{lemma}{Lemma}
\newtheorem{theorem}{Theorem}
\newtheorem{proposition}{Proposition}

\title{Capacitated Network Design Games \\on a Generalized Fair Allocation Model\thanks{This work is partially supported by JSPS KAKENHI Grant Numbers JP20H05967, JP21H05852, JP21K19765, JP21K17707, JP22H00513. The extended abstract of this paper appears in Proceedings of the 21st International Conference on Autonomous Agents and Multiagent Systems
(AAMAS 2022) \cite{HHO2022}.}}

\author[1]{Tesshu Hanaka \thanks{\texttt{hanaka@inf.kyushu-u.ac.jp}}}
\author[2]{Toshiyuki Hirose\thanks{\texttt{ts-hirose@kddi.com}}}
\author[3]{Hirotaka Ono\thanks{\texttt{ono@nagoya-u.jp}}}

\affil[1]{Kyushu University, Fukuoka, Japan}
\affil[2]{KDDI Corporation, Tokyo, Japan}
\affil[3]{Nagoya University, Nagoya, Japan}
\date{}

\begin{document}
\maketitle

\begin{abstract}
\noindent

The cost-sharing connection game is a variant of routing games on a network. 
In this model, given a directed graph with edge costs and edge capacities, each agent wants to construct a path from a source to a sink with low cost. 
The users share the cost of each edge based on a cost-sharing function. 
One of the simple cost-sharing functions is defined as the cost divided by the number of users.
Most of the previous papers about cost-sharing connection games addressed this cost-sharing function.
It models an ideal setting where no overhead arises when people share things, though it might be quite rare in real life; it is more realistic to consider the setting that the cost paid by an agent is the original cost per the number of agents using the edge plus the overhead. 
In this paper, we model the more realistic scenario of cost-sharing connection games by generalizing cost-sharing functions. 
The arguments on the model are based on not concrete cost-sharing functions but cost-sharing functions under a reasonable scheme;  they are applicable for a broad class of cost-sharing functions satisfying the following natural properties: they are (1) non-increasing, (2) lower bounded by the original cost per the number of the agents, and (3) upper bounded by the original cost, which enables to represent various scenarios of cost-sharing. 
We investigate the Price of Anarchy (PoA) and the Price of Stability (PoS) under sum-cost and max-cost criteria with the generalized cost-sharing function. Despite the generalization, we obtain the same tight bounds of PoA and PoS as the cost-sharing with no overhead except PoS under sum-cost. Moreover, for the sum-cost case, the lower bound on PoS increases from $\log n$ to $n+1/n-1$ by the generalization, which is also almost tight because the upper bound is $n$. We further investigate the bounds from the viewpoints of graph classes, such as parallel-link graphs, series-parallel graphs, and directed acyclic graphs, which show critical differences in PoS/PoA values.

\end{abstract}

\section{Introduction}\label{sec:Intro}
The capacitated symmetric cost-sharing connection game (CSCSG) is a network design model of multiple agents' sharing costs to construct a network infrastructure for connecting a given source-sink pair. In the game, a possible network structure is given, but actual links are not built yet. For example, imagine building an overlay network structure on a physical network. Each agent wants to construct a path from source $s$ to sink $t$. To construct a path, each agent builds links by paying the associated costs. Two or more agents can commonly use a link if the number of agents is within the capacity associated with the link, and 
in such a case, the link's cost is fairly shared by the agents that use it. 
Thus, the more agents share one link, the less cost of the link they pay. 
Under this setting, each agent selfishly chooses a path to construct to minimize their costs to pay. 
The CSCSG is quite helpful and can model many real-world situations for sharing the cost of a designed network, such as a virtual overlay, multicast tree, or other sub-network of the Internet. The CSCSG is first introduced by Feldman and Ron~\cite{Feldman2015}. 

In the previous studies, the link cost is fairly shared, 
which means that the total cost paid for a link does not vary even if any number of agents use it.  However, sharing resources yields more or less extra costs (overheads) in realistic cost-sharing situations; by increasing the number of users, extra commission fees are charged, service degradation occurs, and so on. 
The existing models are not powerful enough to handle such situations.


In this paper, we model the more realistic scenario of CSCSG by generalizing cost-sharing functions. 
The arguments on our model do not depend on specific cost-sharing functions, and are applicable to a broad class of cost-sharing functions satisfying certain natural properties. 
Let $p_e$ and $c_e$ be the cost and capacity associated with link (edge) $e$, respectively. Suppose that $x$ agents use link $e$, where $x \le c_e$. In our model,  a cost-sharing function $f_e(x)$ for link $e$ is (1) non-increasing with respect to $x$, 
(2) $f_e(x) \ge p_e/x$, and
(3) $f_e(1) =p_e$. 
Condition (1) is a natural property in cost-sharing models, (2) represents 
the situation that if two or more agents use a link, overheads may arise, and (3) represents that 
no overhead arises when only an agent uses the edge. Note that (2) implies that (2') the total cost paid by all the agents for $e$ is at least $p_e$. Also note that by combining properties (1) and (3), we have (3') $f_e(x) \le p_e$ for any positive integer $x$, which implies that the overhead is not too large and a cost paid by an agent is upper bounded by $p_e$; otherwise no one wants to cooperate. We emphasize that this significant generalization does not restrict any nature of fair cost-sharing. Any natural fair cost-sharing function seems to be in this scheme.  
Note that the cost-sharing function in the previous studies~\cite{Feldman2015,Erlebach2015,Feldman2016} is $f_e(x) = p_e/x$, which satisfies (1), (2) and (3). 



We investigate the Price of Anarchy (PoA) and the Price of Stability (PoS) of the game. 
A pure Nash equilibrium (we simply say Nash equilibrium) is a state where no agent can reduce its cost by changing the path that he/she currently chooses. Such a Nash equilibrium does not always exist in a general game, but it does in CSCSG. 
Thus, a major interest in analyzing games is to measure the goodness of Nash equilibrium. As social goodness measures, we consider two criteria. One is \emph{sum-cost} criterion, where the social cost function is defined as the sum of the costs paid by all the agents, and the other is \emph{max-cost} criterion, which is defined as the maximum among the costs paid by all the agents.
Both PoA and PoS are well-used measures for evaluating the efficiency of Nash equilibria of games. The PoA is the ratio between the cost of the worst Nash equilibrium and the social optimum, whereas the PoS refers to the ratio between the cost of the best Nash equilibrium and the social optimum. 

The previous studies also investigate the PoA and PoS of these cost criteria under their game models. For details, see Section \ref{subsec:related}.  
Despite of the generalization, we
obtain the same bounds of PoA and PoS as the cost-sharing with no overhead except PoS of sum-cost.
Note that these bounds are tight. In the case of sum-cost, the lower bound on PoS increases from $\log n$ to $\PoSLsum$ by the generalization, which is also almost tight because the upper bound is $n$.
We further investigate the bounds from the viewpoints of graph classes, such as parallel-link graphs, series-parallel graphs, and directed acyclic graphs, which show critical differences in PoS/PoA values. The details are summarized in Section \ref{subsec:contribution}.

\begin{table*}[t]
\centering
\caption{The summary of PoA and PoS of CSCSG  under sum-cost criterion.}
\renewcommand{\arraystretch}{1.2}
\resizebox{\textwidth}{!}{ 
\begin{tabular}{cc|c|c|c|c|}
\cline{3-6}
                                                     &                      & \multicolumn{1}{c|}{parallel-link} & \multicolumn{1}{c|}{series-parallel} &
                                                \multicolumn{1}{c|}{DAG} &
                                                \multicolumn{1}{c|}{General}   \\ \hline
                                                     
\multicolumn{1}{|c|}{\multirow{2}{*}{Uncapacitated}} & \multirow{1}{*}{PoA} &           \multicolumn{4}{c|}{$n$ \cite{Anshelevich2008}    }                  \\ \cline{2-6}
\multicolumn{1}{|c|}{}                               & \multirow{1}{*}{PoS} &    \multicolumn{4}{c|}{$1$ (trivial)    }        \\  \hline

\multicolumn{1}{|c|}{\multirow{2}{*}{Capacitated}}   & \multirow{1}{*}{PoA} &    \multicolumn{2}{c|}{$n$  ({UB} \cite{Feldman2015}, LB \cite{Anshelevich2008})}   
& \multicolumn{2}{c|}{\bf unbounded [Thm. \ref{thm:DAG:sum}]}  \\ \cline{2-6} 
\multicolumn{1}{|c|}{}                               
& \multirow{1}{*}{PoS} &          
\multicolumn{4}{c|}{$\log n$ \cite{Feldman2015}\footnotemark[1]}                    \\ \hline

\multicolumn{1}{|c|}{Capacitated+General cost}                               & \multirow{1}{*}{PoA} & 
\multicolumn{2}{c|}{\bf $n$  [Thm. \ref{thm:scpoa}]   } & \multicolumn{2}{c|}{\bf unbounded [Thm. \ref{thm:DAG:sum}]}    \\ \cline{2-6} 
\multicolumn{1}{|c|}{{\bf [Our Setting]}}                    & \multirow{1}{*}{PoS} & \multicolumn{4}{c|}{{\bf UB : $n$}, {\bf LB : $\PoSLsum$ [Thm. \ref{thm:PoS_sc}]}  }         \\ \hline

\end{tabular}
}
\label{sc}
\end{table*}


\begin{table*}[t]
\centering
\caption{The summary of PoA and PoS of CSCSG under max-cost criterion.}
\renewcommand{\arraystretch}{1.2}
\resizebox{\textwidth}{!}{ 
\begin{tabular}{cc|c|c|c|c|}
\cline{3-6}
                                                     &                      & \multicolumn{1}{c|}{parallel-link} & \multicolumn{1}{c|}{series-parallel} &
                                                \multicolumn{1}{c|}{DAG} &
                                                \multicolumn{1}{c|}{General}   \\ \hline
                                                     
\multicolumn{1}{|c|}{\multirow{2}{*}{Uncapacitated}} & \multirow{1}{*}{PoA} &           \multicolumn{4}{c|}{$n$ \cite{Anshelevich2008}   }                  \\ \cline{2-6}
\multicolumn{1}{|c|}{}                               & \multirow{1}{*}{PoS} &    \multicolumn{4}{c|}{$1$ (trivial)    }        \\ \hline
\multicolumn{1}{|c|}{\multirow{2}{*}{Capacitated}}   & \multirow{1}{*}{PoA} &    \multicolumn{2}{c|}{$n$ \cite{Feldman2015}}                   & \multicolumn{2}{c|}{\bf unbounded [Thm. \ref{thm:DAG:max}]}  \\ \cline{2-6} 
\multicolumn{1}{|c|}{}                               
& \multirow{1}{*}{PoS} &   
\multicolumn{4}{c|}{$n$  (UB \cite{Erlebach2015}, LB \cite{Feldman2015}\footnotemark[1])}                      \\ \hline

\multicolumn{1}{|c|}{Capacitated+General cost}                               & \multirow{1}{*}{PoA} & 
\multicolumn{2}{c|}{\bf $n$  [Thm. \ref{thm:mcpoa}]  } & \multicolumn{2}{c|}{\bf unbounded [Thm. \ref{thm:DAG:max}]}    \\ \cline{2-6} 
\multicolumn{1}{|c|}{{\bf [Our Setting]}}                    & \multirow{1}{*}{PoS} & \multicolumn{4}{c|}{\bf $n$  [Thm. \ref{thm:mcpos}]}           \\ \hline

\end{tabular}
}
\label{mc}
\end{table*}

\footnotetext[1]{In~\cite{Feldman2015}, Feldman and Ron gave the lower bounds only for \emph{undirected} parallel-link graphs. They  can be easily modified to \emph{directed} parallel link graphs.
} 


\subsection{Our contribution}\label{subsec:contribution}

In this paper, we investigate the PoA and the PoS of CSCSG under a generalized cost-sharing scheme, as explained above.
We address two criteria of the social cost: sum cost and max cost.


As for the sum-cost case, we first show that PoA is unbounded even on directed acyclic graphs (DAGs). 
On the other hand, on series-parallel graphs (SP graphs), we show that PoA under sum-cost is at most $n$ and it is tight, that is, there is an example whose PoA is $n$. For PoS, we show that it is at most $n$, and there is an example whose PoS under sum-cost is $\PoSLsum$. 
This shows the difference from the previous study, which shows 
that PoS is at most $\log n$ and it is tight with ordinary fair cost-sharing functions~\cite{Feldman2015}.

Next, we give the results on the max-cost.
As with the sum-cost, we show that PoA under max-cost is unbounded on directed acyclic graphs. 
On SP graphs, we prove that PoA is at most $n$ and is tight.
We also show that PoS is at most $n$ and is tight.
These results imply that the significant generalization does not affect PoA and PoS under max-cost.
We summarize the results of CSCSG in Tables \ref{sc} and \ref{mc} for sum-cost and max-cost, respectively.

We then discuss the capacitated \emph{asymmetric} cost-sharing connection game (CACSG),  where agents have different source and sink nodes. We show that the lower bounds of PoA and PoS of CSCSG hold for the asymmetric case, while PoS under sum-cost and max-cost are  at most $n$ and $n^2$, respectively.

Remark that we consider the games on directed graphs, but all the results except for directed acyclic cases can be easily modified to undirected cases. In this sense, our results are generic, which includes the results of \cite{Feldman2015,Erlebach2015}.

\subsection{Related work}\label{subsec:related}
The cost-sharing connection game (CSG) is first introduced by Anshelevich et al. \cite{Anshelevich2008}.
The paper gives the tight bounds of PoA and PoS under sum-cost, which are $n$ and $1$, respectively, for \emph{uncapacitated} CSG.  
They also show that the PoS under sum-cost of \emph{asymmetric} CSG, where agents have different source and sink nodes, can be bounded by $\log n$.
Epstein, Feldman, and Mansour study the  \emph{strong equilibria} of cost-sharing connection games \cite{Epstein2009}.
In \cite{FalkenhausenH13}, von Falkenhausen and Tobias propose cost-sharing protocols with generalized cost-sharing functions.

Feldman and Ron \cite{Feldman2015} introduce a capacitated variant of CSGs on undirected graphs. 
For the variant, they give the tight bounds of PoA and PoS under both sum-cost and max-cost for several graph classes, except the PoS under max-cost for general graphs. 
Note that their results only hold for symmetric CSG.
Erlebach and Radoja fill the gap of the exception \cite{Erlebach2015}. 
Feldman and Ofir investigate the strong equilibria for the capacitated version of CSG \cite{Feldman2016}.

In the literature on computing a Nash equilibrium, 
Anshelevich et al. prove that computing a ``cheap'' Nash equilibrium on CSG is NP-complete  \cite{Anshelevich2008}.
Also, Syrgkanis shows that finding a Nash equilibrium on CSG  is PLS-complete \cite{Syrgkanis2010}.

There are {vast applications} of CSG.
A natural application is a decision making in sharing economy \cite{Anshelevich2008,Chau2018,chau2019}. 
Radko and Laclau mention the relationship between CSG and machine learning \cite{Ievgen2019}. 
The previous studies for CSCSG do not consider any overhead, but when we share some resources (or tasks), it 
generally yields some overhead. In fact, controlling overheads to share  tasks is a major issue in grid/parallel computing fields~\cite{goldman2004model}. Furthermore, in the context of sharing economy, the transaction cost is considered a part of overheads~\cite{henten2016transaction}.

\bigskip 

The rest of this paper is organized as follows. 
Section \ref{sec:Model} is preliminary: we give notations and a formal model of CSCSG, 
and so on. 
Sections \ref{sec:SumCost} and \ref{sec:MaxCost} are the main parts of this paper. The first is about PoA and PoS under sum cost, and the second is about PoA and PoS under max cost. Section \ref{sec:asymmetric} discusses the capacitated asymmetric CSG. 

\section{Model}\label{sec:Model}
\subsection{Capacitated Symmetric Cost-Sharing Connection Games}

A capacitated symmetric cost-sharing connection game (CSCSG) $\Delta$ is a tuple: $\Delta = (N, G=(V, E), s, t, \{p_e\}_{e\in E}, \{c_e\}_{e\in E})$
where $N$ is the set of agents and $n=|N|$, 
$G=(V, E)$ is a directed graph, $s,t\in V$ are the \emph{source} and \emph{sink} nodes, $p_e\in \mathbb{R}^{\ge 0}$ is the cost of an edge $e$, and $c_e\in\mathbb{N}^{\ge 0}$ is the capacity of an edge $e$, i.e., the upper bound of the number of agents that can use edge $e$. An edge is also called a \emph{link}. 
The purpose of each agent $j$ is to construct an $s$-$t$ path in $G$. An $s$-$t$ path chosen by agent $j$ is called a \emph{strategy} of agent $j$, denoted by $s_j$. 
A tuple $\s = (s_1, \ldots, s_n)$ of strategies of $n$ agents is called a \emph{strategy profile}.
We denote by $E(s_j)\subseteq E$ the set of edges in $s$-$t$ path $s_j$. Namely, $E(s_j)$ is the set of edges used by agent $j$.
Moreover, for a strategy profile $\s = (s_1, \ldots, s_n)$, we define $E(\s)=\bigcup_j E(s_j)$, which is the set of edges used in $\s$.
%
%
Let $x_e(\s)=|\{j \mid e\in E(s_j)\}|$ be the number of agents who use edge $e$ in a strategy profile $\s$. 
A strategy profile $\s$ is said to be \emph{feasible} if $\s$ satisfies $x_e(\s)\le c_e$ for every $e\in E$. 
Furthermore, we say the game is \emph{feasible} if at least one feasible strategy profile exists. 
In this paper, we deal with only feasible games, that is, games have at least one feasible strategy profile.





\subsection{Cost-Sharing Function and Social Cost}

In a general network design game, for a strategy profile $\s$, every agent $j$ who uses an edge $e$ should pay some cost based on a \emph{payment} function $f_{e,j}(x_e(\s))$; agent $j$ pays $\sum_{e\in E(s_j)}f_{e,j}(x_e(\s))$ in total. In cost-sharing connection games, the cost imposed to an edge $e$ is \emph{fairly} divided into the agents using $e$; the cost paid by an agent using $e$ is determined by a \emph{cost-sharing function} $f_e(x_e(\s))$, and the total cost of agent $j$ is     
\begin{align*}
 p_j(\s)=
 \begin{cases}
  \sum_{e\in E(s_j)}f_e(x_e(\textbf{s})) & \forall e\in E(s_j), x_e(\textbf{s}) \le c_e \\ 
  \infty & \mbox{otherwise}
 \end{cases}.
\end{align*}
In this paper, we assume that a cost-sharing function $f_e(x)$ satisfies (1) \emph{non-increasing}, (2) $f_e(x)\ge p_e/x$ and (3) $f_e(1)=p_e$.  

We denote by $(\Delta, F)$ a CSCSG on $\Delta$ with the set of cost-sharing functions $F=\{f_e \mid e\in E\}$. To emphasize that all functions in $F$ satisfy the properties (1), (2) and (3), we say that $F$ is in the generalized cost-sharing scheme, denoted by $\mathcal{F^*}$. Recall that (2) implies (2’) the total cost paid by all the agents for $e$ is at least $p_e$, and (3’) $f_e(x) \le p_e$ for any $x\ge 1$. If there is no overhead for sharing an edge $e$, the cost agent $j$ pays for $e$ under $\textbf{s}$ is defined by $f_e(x_e(\textbf{s})) = p_e/x_e(\textbf{s})$. Let us denote $F_{\ord}=\{p_e/x_e(\textbf{s}) \mid e\in E\}$, and clearly $F_{\ord}\in \mathcal{F}^*$. Previous studies such as \cite{Anshelevich2008,Epstein2009,Erlebach2015,Feldman2016,Feldman2015} adopt this special case $(\Delta, F_{\ord})$. 

We further denote by $\mathcal{F}_{all}$ the class of any payment functions, including non-fair or even meaningless ones in the cost-sharing context. We introduce this class of functions to contrast it with $\mathcal{F}^*$. For a class $\mathcal{F}$ of cost-sharing functions, we sometimes write $(\Delta, \mathcal{F})$ instead of writing ``$(\Delta, F)$ for any $F\in \mathcal{F}$''.




In CSCSG, we consider two {types of} social costs for strategy profiles.
The \emph{sum-cost} of a strategy profile $\s$ is the total cost of all agents, that is, $\cost_{sc}(\textbf{s})=\sum_j p_j(\textbf{s})$.  
The \emph{max-cost} of a strategy profile $\s$ is the maximum {among the costs paid by all the agents}, that is, $\cost_{mc}(\textbf{s})=\max_j p_j(\textbf{s})$.


\subsection{The Existence of Nash Equilibrium}
Given a strategy profile \textbf{s},  let $s'_j$ be a new strategy of agent $j$ and $\textbf{s}_{-j}=\textbf{s}\setminus \{s_j\}$
be the strategy profile $\textbf{s}$ 
excluding $s_j$.
If there is an agent $j$ such that $p_j(\textbf{s})> p_j(s'_j,\ \textbf{s}_{-j})$ for some $s'_j$, agent $j$ has an incentive to change its strategy from $s_j$ to $s'_j$. We call this type of change a \emph{deviation}.  
A strategy profile \textbf{s} is called a \emph{Nash equilibrium} if any agent 
does not have an incentive to deviate from \textbf{s}, that is,
$p_j(\textbf{s})\le p_j(s'_j,\ \textbf{s}_{-j})$
holds for any agent $j$ and any $s'_j$.
We denote the set of  Nash equilibria in {\rm CSCSG} $(\Delta, F)$ by $\NE(\Delta, F)$.

The proof of Theorem 2.1 of \cite{Anshelevich2008} shows that any non-capacitated network design game always has a Nash equilibrium by an argument using a potential function. Because we only consider feasible games,  a similar argument can be applied to {\rm CSCSG} $(\Delta, F)$ using the following potential function:
$\Phi(\textbf{s})=\sum_{e\in E}\sum_{x=1}^{x_e(\s)}f_e(x)$, where $f_e \in F$. 
\begin{proposition}
\label{prop:potential}
For any {\rm CSCSG} $(\Delta, \mathcal{F}_{all})$, there exists a pure Nash equilibrium.
\end{proposition}

\subsection{Price of Anarchy and Price of Stability}
The Price of Anarchy (PoA) and the Price of Stability (PoS)  measure how inefficient the cost of a Nash equilibrium is compared to an optimal cost. 
Let \textbf{s}$^*_{sc}$ be an optimal strategy profile under sum-cost, and 
\textbf{s}$^*_{mc}$ be an optimal strategy profile under max-cost, respectively. 
Then the PoA's of $(\Delta,F)$ under sum-cost and max-cost are defined as follows.
\begin{align*}
\PoA_{sc}(\Delta,F)=\frac{\max_{\textbf{s}\in {\rm NE}(\Delta,F)}\cost_{sc}(\textbf{s})}{\cost_{sc}(\textbf{s}^*_{sc})}, \\
\PoA_{mc}(\Delta,F)=\frac{\max_{\textbf{s}\in {\rm NE}(\Delta,F)}\cost_{mc}(\textbf{s})}{\cost_{mc}(\textbf{s}^*_{mc})}
\end{align*}
Similarly, the PoS's of $(\Delta,F)$ under sum-cost and max-cost are defined as follows.
\begin{align*}
\PoS_{sc}(\Delta,F) = \frac{\min_{\textbf{s}\in {\rm NE}(\Delta,F)}\cost_{sc}(\textbf{s})}{\cost_{sc}(\textbf{s}^*_{sc})}, \\
\PoS_{mc}(\Delta,F) = \frac{\min_{\textbf{s}\in {\rm NE}(\Delta,F)}\cost_{mc}(\textbf{s})}{\cost_{mc}(\textbf{s}^*_{mc})}
\end{align*}
When it is clear from the context, we sometimes omit $(\Delta,F)$.



\subsection{Graph Classes}

A \emph{single source single sink directed acyclic graph} is a directed graph with exactly one source node $s$ and sink node $t$ and without cycles.
We simply call it a directed acyclic graph (DAG) in this paper.


A \emph{two-terminal series-parallel graph} $G$ is a directed graph with  exactly one source node $s$ and sink node $t$ that can be constructed by a sequence of the following three operations  \cite{epstein}: 
(i)  Create a single directed edge $(s,t)$. 
(ii) Given two two-terminal series-parallel graphs $G_X$ with terminals $s_X$ and $t_X$ and $G_Y$ with terminals $s_Y$ and $t_Y$, form a new graph $S(G_X,G_Y)$ with terminals $s$ and $t$ by identifying $s=s_X$, $t_X=s_Y$ and $t=t_Y$. We call this operation the \emph{series composition} of $X$ and $Y$. 
(iii) Given two two-terminal series-parallel graphs $G_X$ with terminals $s_X$ and $t_X$ and $G_Y$ with terminals $s_Y$ and $t_Y$, form a new graph $P(G_X,G_Y)$ with terminals $s$ and $t$ by identifying $s=s_X=s_Y$ and $t=t_X=t_Y$. We call this operation the \emph{parallel composition} of $G_X$ and $G_Y$.

Note that any two-terminal series-parallel graph is a directed acyclic graph.
We call a two-terminal series-parallel graph a series-parallel graph (SP graph) for simplicity \cite{epstein}.
An SP graph $G$ is a \emph{parallel-link graph} if it is produced by only parallel compositions of single edges. 

By the definitions of the above graphs, the following inclusion relation holds, where each name represents a graph class: $\mbox{Parallel-link graph}\subseteq \mbox{SP graph} \subseteq \mbox{DAG} \subseteq \mbox{General graph}$.

\section{Capacitated Cost-Sharing Connection Games under Sum-Cost Criterion}\label{sec:SumCost}
In this section, we give the bounds of PoA and PoS of {\rm CSCSG} under the sum-cost criterion.

\subsection{Price of Anarchy (PoA)}

In \cite{Feldman2015}, Feldman and Ron show that $\PoA_{sc}$ of CSCSG is unbounded on \emph{undirected} graphs when the cost-sharing functions is in $F_{\ord}$, that is, $f_e(x_e(\s))=p_e/x_e(\s)$ for each $e\in E$. Because any CSCSG on an undirected graph can be transformed into a CSCSG on a directed graph \cite{Erlebach2015}, $\PoA_{sc}$ is unbounded on directed graphs. However, the transformation yields directed cycles.

\subsubsection{Directed acyclic graphs}
In this subsection, we show that $\PoA_{sc}$ of $(\Delta, F_{\ord})$ is unbounded even on directed acyclic graphs (DAG). 
In the proof, we use a directed acyclic graph shown in Figure \ref{Fig:DAG} whose edge costs are represented by some variables. We can show that $\PoA_{sc}$ can be infinitely large by controlling the variables. 
\begin{theorem}\label{thm:DAG:sum}
There exists a CSCSG $(\Delta, F_{\ord})$  on DAGs such that $\PoA_{sc}$ is unbounded. 
\end{theorem}

\begin{proof}
We give an example with two agents such that $\PoA_{sc}$ is unbounded where $\alpha<\beta$ (see Figure \ref{Fig:DAG}). 
Suppose one agent uses path $s\rightarrow a \rightarrow c \rightarrow t$, and the other uses path $s\rightarrow b \rightarrow t$. Then the sum-cost of this strategy profile is $5\alpha$. 

On the other hand, consider a strategy profile such that one agent uses path $s\rightarrow a \rightarrow b \rightarrow t$ and the other agent uses path $s\rightarrow b\rightarrow c\rightarrow t$. 
It is easy to see that this strategy profile is a Nash equilibrium and 
its sum cost is $4\alpha+2\beta$. 
Thus, the $\PoA_{sc}$ is at least $(4\alpha+2\beta)/5\alpha=4/5+2\beta/5\alpha$.
{By taking $\beta=\alpha^2$ and $\alpha$ arbitrary large}, the $\PoA_{sc}$ can be unbounded.
\end{proof}

\begin{figure}[tbp]
  \begin{center}
  \includegraphics[width=0.7\linewidth]{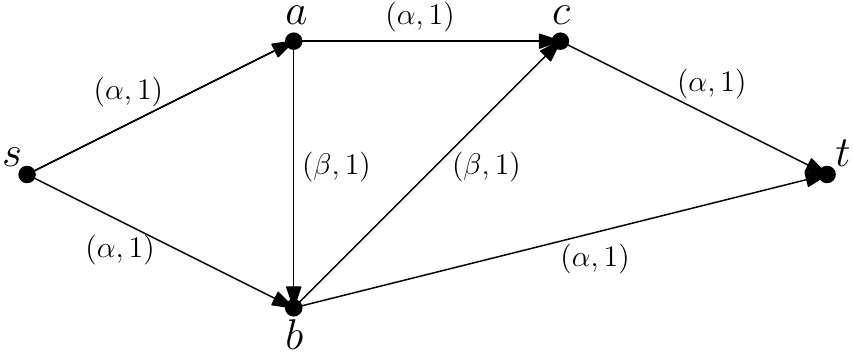}
  \end{center}
  \caption{A CSCSG $(\Delta, \{p_e/x_e\mid e\in E\})$ on DAGs such that $\PoA_{sc}$ is unbounded. The cost and capacity of an edge $e$ is denoted by $(p_e, c_e)$.}
  \label{Fig:DAG}
\end{figure}

\subsubsection{SP graphs}
For SP graphs, we show that $\PoA_{sc}$ is at most $n$, and it is tight. 
In \cite{Feldman2015}, Feldman and Ron show that the upper bound of PoA on (undirected) SP graphs is $n$ for the ordinary cost-sharing functions. We claim that it actually holds for \emph{directed} SP graphs and any cost-sharing functions in $\mathcal{F}^*$. We first introduce Lemma \ref{lem:feasible} shown in Theorem 6 in \cite{Feldman2015}, which holds 
for \emph{any} payment functions. Although the lemma is stated in the game-theoretic context, the claim is essentially about the network flow. Note that the original proof is for undirected SP graphs, but it can be easily modified to directed cases, though we omit the detail.

\begin{lemma}[\rm Lemma 6 in \cite{Feldman2015}]\label{lem:feasible}
Let $(\Delta, \mathcal{F}_{all})$ be a CSCSG on SP graphs. 
For $r, k \in \mathbb{N}$ where $r<k$, let \textbf{s} be a feasible strategy profile of $k$ agents, and $\textbf{s}'$ be a feasible strategy profile of $r$ agents. 
Then, there is an $s$-$t$ path $s_{r+1}$ in $G$ that uses only edges used in $\textbf{s}$ such that the strategy profile $(\textbf{s}',s_{r+1})$ of $r+1$ agents is feasible.
\end{lemma}

We obtain the following lemma by using Lemma \ref{lem:feasible}. Note that Lemma \ref{lem:feasible} holds for any payment functions, but Lemma \ref{lem:cost} holds only for cost-sharing functions in the generalized cost-sharing scheme of this paper.  

\begin{lemma}
\label{lem:cost}
Let $(\Delta, \mathcal{F}^*)$ be a CSCSG on SP graphs, and let  $\textbf{s}^*$ be an optimal strategy profile under sum-cost and \textbf{s} be a strategy profile that is a Nash equilibrium in $(\Delta, \mathcal{F}^*)$.
Then, the cost of each agent in $\s$ is at most $\cost_{sc}(\s^*)$.
\end{lemma}

\begin{proof}
Let $\textbf{s}^*$ be an optimal strategy profile under the sum-cost criterion, and \textbf{s} be a Nash equilibrium strategy profile.
Then we show that the cost of each agent in $\s$ is at most $\cost_{sc}(\s^*)$.
For the sake of contradiction, we assume that there is an agent $i$ whose cost $p_i(\s)$ is higher than $\cost_{sc}(\textbf{s}^*)$. 
Let \textbf{s}$_{-i}$ be the strategy profile for all agents except for agent $i$ and $s_i$ be the $s$-$t$ path chosen by agent $i$ in $\s$. 
Given the strategy profile \textbf{s}$_{-i}$, there is a feasible $s$-$t$ path $s'$ that uses only edges in $\textbf{s}^*$ by Lemma \ref{lem:feasible}. 
If agent $i$ chooses the $s$-$t$ path $s'$ instead of $s_i$, we claim that the cost of agent $i$ becomes at most $\cost_{sc}(\s^*)$. This can be shown as follows. 
In the original strategy $\s$, agents using edge $e$ pay $f_e(x_e(\s))$ for each, and  
in the new strategy $(\s_{-i}, s')$, agent $i$ needs to pay $f_e(x_e(\s) + 1)$ for $e\in E(s')\setminus E(s_i)$. By taking the summation, the total cost that agent $i$ pays in (\textbf{s}$_{-i}, s')$ is 
\begin{align*}
\sum_{e\in E(s')\cap E(s_i)}f_e(x_e) + \sum_{e\in E(s')\setminus E(s_i)}f_e(x_e+1) 
&\le  \sum_{e\in E(s')\cap E(s_i)}p_e  + \sum_{e\in E(s')\setminus E(s_i)} p_e \\
&= \sum_{e\in E(s')}p_e \\
&\le  \sum_{e\in E(\mathbf{s}^*)} p_e \\
&\le \cost_{sc}(\s^*).
\end{align*}
The first and last inequalities come from properties (3') and (2') of our generalized cost-sharing scheme, respectively. 
Thus, agent $i$ can pay less by deviating to this path.  
This contradicts  that the strategy profile \textbf{s} is a Nash equilibrium. 
Thus, $p_j(\textbf{s})\le \cost_{sc}(\s^*)$ for any agent $j$.
\end{proof}

By Lemma \ref{lem:cost}, we can see that the total cost of the agents in a Nash equilibrium is at most $n\cdot \cost_{sc}(\s^*)$, which implies the following.  

\begin{lemma}\label{lem:ubscpoa}
In CSCSG $(\Delta, \mathcal{F}^*)$ on SP graphs, 
$\PoA_{sc}$ is at most $n$.
\end{lemma}

As for the lower bound of $\PoA_{sc}$, Anshelevich et al. give an example of \emph{uncapacitated} cost-sharing connection games on parallel-link graphs where $\PoA_{sc}$ is $n$  \cite{Anshelevich2008}. The example is a game of $n$ agents on a parallel-link graph consisting of two vertices and two directed edges whose costs  are defined by $1$ and $n$, respectively. 
Since a CSCSG such that the capacity of each edge is $n$ is equivalent to an uncapacitated cost-sharing connection game, we obtain the same lower bound for CSCSG.


\begin{lemma}
\label{lem:lbscpoa}
There exists a CSCSG $(\Delta, F_{\ord})$ where $\PoA_{sc}$ is $n$ even on parallel-link graphs.
\end{lemma}
By Lemmas \ref{lem:ubscpoa} and \ref{lem:lbscpoa}, we obtain  Theorem \ref{thm:scpoa}.

\begin{theorem}\label{thm:scpoa}
For any CSCSG $(\Delta, \mathcal{F}^*)$ on SP graphs, $\PoA_{sc}$ is at most $n$. Furthermore, there exists a CSCSG with $\PoA_{sc}=n$ on parallel-link graphs.
\end{theorem}

\subsection{Price of Stability (PoS)}\label{sec:PoS_sum}

In this subsection, we show that for any CSCSG $(\Delta, \mathcal{F}^*)$, $\PoS_{sc}$  is at most $n$, and it is almost tight.
{We first show the following property obtained from 
the non-increasingness of $f_e\in \mathcal{F}^*$.}

\begin{lemma}\label{lem:non-increasing}
$\cost_{sc}(\textbf{s}) \le \Phi(\textbf{s})$ holds.
\end{lemma}

\begin{proof}

Recall that $\Phi(\textbf{s})=\sum_{e\in E}\sum_{x=1}^{x_e(\s)}f_e(x)$.
Considering that $f_e\in \mathcal{F}^*$ is non-increasing, we obtain:
\begin{align*}
\Phi(\textbf{s})&=\sum_{e\in E}\sum_{x=1}^{x_e(\s)}f_e(x)\\
&\ge \sum_{e\in E}x_e(\s)f_e(x_e(\s))\\
&= \sum_{e\in E(\s)}x_e(\s)f_e(x_e(\s))\\
&=\sum_{j=1}^n\sum_{e\in E(s_j)}f_e(x_e(\s))\\
&=\cost_{sc}(\textbf{s}).
\end{align*}
\end{proof}

{Then we show the upper bound of $\PoS_{sc}$.}

\begin{lemma}\label{lem:ubscpos}
For any CSCSG $(\Delta, \mathcal{F}^*)$, $\PoS_{sc}$ is at most $n$.  
\end{lemma}

\begin{proof}
Let $\s^*$ be an optimal strategy profile under sum-cost.
Consider agents repeatedly deviate from $\s^*$ to reduce their costs. 
Eventually, this procedure results in a Nash equilibrium $\s$  by the proof of Proposition \ref{prop:potential}.

Recall that the change $\Phi(\s)-\Phi(s'_j,\s_{-j})$ from $\s$ to a new strategy profile $(s'_j,\s_{-j})$ equals the change of the cost of agent $j$ \cite{Monderer1996}.
Thus,  $\Phi(\s)\le\Phi(\s^*)$ holds. 

By property (3’)  of  our  generalized  cost-sharing  scheme, for any edge $e\in E$ and strategy profile $\s$, $\sum_{x=1}^{x_e(\s)}f_e(x) \le {np_e}$ holds.
Then we transform the potential function $\Phi(\s^*)$ for  strategy profile $\s^*$ as follows.

\begin{align}\label{inequality:max:bound}
    \Phi(\textbf{s}^*)&= \sum_{e\in E}\sum_{x=1}^{x_e(\textbf{s}^*)}f_e(x) \nonumber \\
    &= \sum_{e\in E(\s^*)}\sum_{x=1}^{x_e(\textbf{s}^*)}f_e(x)  \\
   & \le \sum_{e \in E(\s^*)}np_e \nonumber  \\
   &\le n\cdot \cost_{sc}(\s^*) \nonumber 
\end{align}

{By Lemma \ref{lem:non-increasing}}, we have
    $\cost_{sc}(\textbf{s})\le\Phi(\textbf{s})\le\Phi(\textbf{s}^*)\le  n\cdot \cost_{sc}(\textbf{s}^*)$.
Therefore, 
    $\PoS_{sc}(\Delta, \mathcal{F}^*) \le {n\cdot \cost_{sc}(\textbf{s}^*)}/{\cost_{sc}(\textbf{s}^*)}= n$.
\end{proof}

\begin{lemma}\label{lem:lbscpos}
There exists a CSCSG $(\Delta, \mathcal{F}^*)$ with $\PoS_{sc}=\PoSLsum$ on parallel-link graphs. 
\end{lemma}

\begin{figure}[t]
 \centering
 \includegraphics[width=0.4\linewidth]{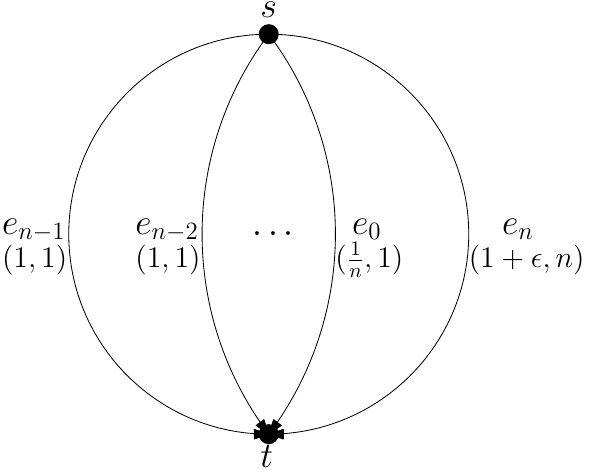}
 \caption{A CSCSG $(\Delta, \mathcal{F}^*)$ on parallel-link graphs with $\PoS_{sc}=\PoSLsum$. The cost and capacity of edge $e$ is denoted by $(p_e, c_e)$.
 \label{fig3}}
\end{figure}

\begin{proof}
Consider the following CSCSG $(\Delta, \mathcal{F}^*)$ with $n$ agents on the parallel-link graph with $n+1$ edges $e_0, \ldots, e_n$, illustrated in Figure \ref{fig3}. 
For edges $e_0$ and $e_n$, we define their costs and capacities by $(p_{e_0},c_{e_0})=(1/n, 1)$ and  
$(p_{e_n},c_{e_n})=(1+\epsilon, n)$. 
Also, the cost and capacity of edge $e_i$ are defined by $(p_{e_i},c_{e_i})=(1, 1)$  for $1\le i\le n-1$. 
Finally, we define the cost-sharing function of edge $e$ as follows:
\begin{align*}
    f_e(x)=\begin{cases}
    p_e & (x\neq n)\\
    p_e/x & (x = n)
    \end{cases}.
\end{align*}
Note that $f_e$ satisfies $f_e(1)=p_e$ and $f_e(x)\ge p_e/x$ 
and thus $f_e$'s belong to $\mathcal{F}^*$.

Let $\s^*$ be a strategy profile where all the agents use $e_n$.
The sum-cost of $\s^*$ is $\cost_{sc}(\s^*)=(1+\epsilon)/n\cdot n=1+\epsilon$. 
Note that $\s^*$ is not a  Nash equilibrium because an agent can reduce the cost from $(1+\epsilon)/n$ to $1/n$ by moving from $e_n$ to $e_0$.

Next, let $\s$ be a strategy profile where agent $i$ uses edge $e_{i-1}$ where $1\le i \le n$.
It is easy to check that $\s$ is a Nash equilibrium because no agent has an incentive to move to $e_n$. 
We then show that no other strategy profile is a Nash equilibrium by contradiction. 
Suppose that $\s'$ is another Nash equilibrium, where at least one agent must use $e_n$. 
The number of agents using $e_n$ is at most $n-1$ because an agent must use $e_0$ of the smallest cost.   
This implies that the cost of an agent using $e_n$ is $1+\epsilon$ by the definition of $f_e$. Then, 
an agent using $e_n$ can reduce the cost from $1+\epsilon$ to $1$ by moving from $e_n$ to an empty edge between $e_1$ to $e_{n-1}$ (such an edge must exist); $\s'$ cannot be a Nash equilibrium. 
This shows that $\s$ is the unique Nash equilibrium. 
The sum cost in strategy profile $\textbf{s}$ is $\cost_{sc}(\s)=n-1+1/n$.


Therefore,
    $\PoS_{sc}(\Delta, \mathcal{F}^*)
    \ge \cost_{sc}(\textbf{s})/\cost_{sc}(\textbf{s}^*)
    =(n-1+{1}/{n})/(1+\epsilon)$.
When $\epsilon$ is arbitrarily small,
$\PoS_{sc}(\Delta, \mathcal{F}^*)$ becomes $\PoSLsum$.
\end{proof}

\begin{theorem}\label{thm:PoS_sc}
For any  CSCSG $(\Delta, \mathcal{F}^*)$, 
$\PoS_{sc}$ is at most $n$. 
Furthermore, there is a CSCSG $(\Delta, \mathcal{F}^*)$ with $\PoS_{sc}=\PoSLsum$ even on parallel-link graphs.
\end{theorem}

\section{Capacitated Cost-Sharing Connection Games under Max-Cost Criterion}\label{sec:MaxCost}
In this section, we give the tight bounds of PoA and PoS of {\rm CSCSG} under max-cost.

\subsection{Price of Anarchy (PoA)}
Feldman et al. show that $\PoA_{mc}$ is unbounded on  undirected graphs  \cite{Feldman2015}. 
As with the sum-cost case, we can show that $\PoA_{mc}$ is unbounded on directed graphs by transforming from an  undirected graph to a directed graph in \cite{Erlebach2015}. 
However, recall that the transformation yields directed cycles; hence we give stronger results on restricted graph classes.








\subsubsection{Directed acyclic graphs}
We show that $\PoA_{mc}$ of a CSCSG $(\Delta, F_{\ord})$ is unbounded even on DAGs.
More precisely, $\PoA_{mc}$ of the CSCSG with two agents illustrated in Figure \ref{Fig:DAG} is unbounded.

\begin{theorem}\label{thm:DAG:max}
There exists a CSCSG $(\Delta, F_{\ord})$ on DAGs such that $\PoA_{mc}$ is unbounded. 
\end{theorem}

\begin{proof}
We show that $\PoA_{mc}$ of the CSCSG with two agents illustrated in Figure \ref{Fig:DAG} is unbounded.
The strategy is the same as the sum-cost case. 
Consider that one agent uses path $s\rightarrow a \rightarrow c \rightarrow t$, and the other uses path $s\rightarrow b \rightarrow t$. Then the max cost of this strategy profile is $3x$.

On the other hand, consider a Nash equilibrium such that one agent uses path $s\rightarrow a \rightarrow b \rightarrow t$ and the other agent uses path $s\rightarrow b\rightarrow c\rightarrow t$. Then its max-cost is $2x+y$. Thus, we have $\PoA_{mc}\ge (2x+y)/3x=2/3+y/3x$. Because $x$ and $y$ are arbitrary where $x<y$, $\PoA_{mc}$  can be unbounded.
\end{proof}

\subsubsection{SP graphs}
For any CSCSG $(\Delta, \mathcal{F}^*)$ on SP graphs, we show that $\PoA_{mc}$  is at most $n$, and it is tight.

\begin{lemma}\label{lem:ubmcpoa}
For any CSCSG $(\Delta, \mathcal{F}^*)$ on SP graphs, $\PoA_{mc}$ is at most $n$.
\end{lemma}
\begin{proof}

Let $\s^*$ be an optimal strategy profile under the max cost criterion, and $\s$ be a Nash equilibrium.
By Lemma \ref{lem:cost}, the cost of each agent in $\s$ is at most $\cost_{sc}(\textbf{s}^{**})$ where $\textbf{s}^{**}$ is an optimal strategy profile under sum-cost. Thus, $\cost_{mc}(\textbf{s})\le \cost_{sc}(\textbf{s}^{**})$ holds. By the optimality of $\textbf{s}^{**}$ under sum-cost, we have $\cost_{sc}(\textbf{s}^{**})\le \cost_{sc}(\textbf{s}^*)$. Finally, by the definition of max-cost,  $\cost_{sc}(\textbf{s}^*)
    \le n\cdot \cost_{mc}(\textbf{s}^*)$ holds. Summarizing these inequalities, we obtain $\cost_{mc}(\textbf{s})\le n\cdot \cost_{mc}(\textbf{s}^*)$.
Thus, 
we have
  $\PoA_{mc}(\Delta,\mathcal{F}^*)
  \le {n\cdot \cost_{mc}(\textbf{s}^*)}/{\cost_{mc}(\textbf{s}^*)}
  =n$.
\end{proof}

On the other hand, we observe that $\PoA_{mc}$ of the game used in Lemma \ref{lem:lbscpoa} is $n$.
Therefore, we obtain Theorem \ref{thm:mcpoa} as follows. 
\begin{theorem}\label{thm:mcpoa}
For any CSCSG $(\Delta, \mathcal{F}^*)$ on SP graphs, $\PoA_{mc}$ is at most $n$.
Furthermore, there is a CSCSG with $\PoA_{mc}=n$ even on parallel-link graphs. 
\end{theorem}

\subsection{Price of Stability (PoS)}

For the lower bound of $\PoS_{mc}$, Feldman and Ron show that there is a CSCSG instance on an undirected parallel-link graph with two vertices such that $\PoS_{mc}$ is $n$  \cite{Feldman2015}.
By orienting all the edges from a vertex to the other vertex in the undirected parallel-link graph of the CSCSG, we have the same lower bound of $\PoS_{mc}=n$ on directed parallel-link graphs. 
In the following theorem, we then give a tight upper bound to $\PoS_{mc}$.

\begin{theorem}\label{thm:mcpos}
For any CSCSG $(\Delta, \mathcal{F}^*)$, $\PoS_{mc}$ is at most $n$, and it is tight. 
\end{theorem}

\begin{proof}
The outline of the proof follows that of \cite[Theorem 3]{Erlebach2015}, though we extend it to our generalized setting. 

Let $\s^*$ be an optimal strategy profile under max-cost.
As with Lemma \ref{lem:ubscpos}, consider that agents repeatedly deviate from $\s^*$ to reduce their costs. 
By Proposition \ref{prop:potential}, we obtain a Nash equilibrium $\s$ at the end of the deviations.
Without loss of generality, we can scale the edge costs so that $\cost_{sc}(\s^*)=n$, and as a result, we have $\cost_{mc}(\s^*)\ge 1$. 

If $\cost_{mc}(\textbf{s})\le n$, then we obtain $\PoS_{mc}(\Delta, \mathcal{F}^*)\le n$.
Otherwise, the following inequality holds: 
    $n < \cost_{mc}(\textbf{s}) \le \cost_{sc}(\textbf{s})\le \Phi(\textbf{s})<\Phi(\textbf{s}^*)$.
Note that the third inequality comes from Lemma \ref{lem:non-increasing} and the fourth inequality comes from the definition of the potential function.
For some $\alpha, \beta, \delta > 0$, let $\Phi(\textbf{s}^*)=n+\alpha$, $\cost_{mc}(\textbf{s})=n+\beta$, and $\Phi(\textbf{s})=\Phi(\textbf{s}^*)-\delta$. 
Note that $0<\beta\le\alpha-\delta$. 

Let $\s_{-i}$ be the partial strategy profile of $n-1$ agents except for agent $i$ who pays $\cost_{mc}(\s)$. 
Because the change of the potential function equals the change of the cost of an agent who deviates, we have $\Phi(\s_{-i})=\Phi(\s) -\cost_{mc}(\textbf{s})$.

We construct a strategy profile $\s'$ of $n$ agents \textcolor{black}{based on} $\s^*$ and $\s_{-i}$ as follows. We define $\hat{G}=(V, \hat{E})$  where  $\hat{E}=E(\s^*)\cup E(\s_{-i})$ and $\hat{c}(e)=\max\{x_e(\textbf{s}^*), x_e(\textbf{s}_{-i})\}$ as a directed and capacitated graph. 
Now, we have $n-1$  paths in the strategy profile $\s_{-i}$ and $n$ paths in the strategy profile $\s^*$. 
Here, we regard the strategy profile $\s_{-i}$ as an $s$-$t$ flow and let $\hat{G}_{\s_{-i}}$ be the residual network of $\hat{G}$ for $\s_{-i}$.
Since $\hat{G}$ admits a flow of value $n$ while the value of the flow constructed from $\s_{-i}$ is $n-1$,   
there is an augmenting $s$-$t$ path \textcolor{black}{$a$} in $\hat{G}_{\s_{-i}}$. 
\textcolor{black}{Let $\s'$ be the resultant flow of value $n$ by the augmenting $s$-$t$ path $a$ and $F$ be the set of edges in $\hat{E}$ whose users increase in $\s'$.  
We regard $\s'$ as the strategy profile of $n$ agents by decomposing it into $n$ paths.}


\textcolor{black}{By the definition of the potential function $\Phi(\textbf{s})=\sum_{e\in E}\sum_{x=1}^{x_e(\s)}f_e(x)$ and property (3') of $\mathcal{F}^*$, we have: 
\begin{align*}
    \Phi(\s')\le \Phi(\s_{-i}) + \sum_{e\in F} f_e(x_e(\s_{-i})+1) \le \Phi(\s_{-i}) + \sum_{e\in F} p_e. 
\end{align*}
Note that the number of users of each edge in $F$ increases by at most $1$.}

\textcolor{black}{
We then observe that $x_e(\s^*)>0$ holds for every edge $e\in F$ because each edge $e$ in $E(\s_{-i})$ has already been used by $x_e(\s_{-i})$ agents. Thus, we have $\sum_{e\in F} p_e\le \cost_{sc}(\s^*)=n$ by property (2') of $\mathcal{F}^*$.
Since  $\Phi(\s_{-i})=\Phi(\s) -\cost_{mc}(\s)$ and $\cost_{mc}(\s)=n+\beta$, we have: 
    \begin{align*}
    \Phi(\s') &\le \Phi(\s_{-i})+\sum_{e\in F} p_e\\
    &\le \Phi(\textbf{s}_{-i})+n\\
    &=\Phi(\textbf{s})- \cost_{mc}(\s)+n\\
    &=\Phi(\textbf{s})- (n+\beta)+n\\
    &=\Phi(\textbf{s})-\beta\\
    &<\Phi(\textbf{s}).
    \end{align*}
}

Let $\s''$ be a Nash equilibrium obtained from $\s'$ by the deviations of agents.
Then $\Phi(\textbf{s}'')\le\Phi(\textbf{s}')<\Phi(\textbf{s})$.

If $\cost_{mc}(\textbf{s}'')\le n$, then $\PoS_{mc}(\Delta, \mathcal{F}^*)\le n$ due to $\cost_{mc}(\s^*)\ge 1$. Thus, we are done.
Otherwise, we repeat the above procedure starting from $\s''$ instead of $\s$. 
For each time, 
if we obtain a Nash equilibrium under max-cost higher than $n$, it has strictly less potential than the previous Nash equilibrium. 
Since the number of strategy profiles is finite, we eventually obtain a Nash equilibrium whose cost is at most $n$. 
Therefore, $\PoS_{mc}(\Delta, \mathcal{F}^*)\le n$ holds.
\end{proof}

\color{black}

\section{Asymmetric Games}\label{sec:asymmetric}

This section considers the \emph{capacitated asymmetric cost-sharing connection game (CACSG)}, where agents have different source and sink nodes. 
Because the CACSG is a generalization of CSCSG, the lower bound of CSCSG holds for the CACSG. 
\begin{theorem}
For any CACSG $(\Delta, \mathcal{F}^*)$, $\PoA_{sc}$ and $\PoA_{mc}$ are unbounded even on DAGs.
\end{theorem}
As for the sum-cost criterion, it is easily seen that Lemma \ref{lem:ubscpos} holds for the asymmetric case.
\begin{theorem}
For any CACSG $(\Delta, \mathcal{F}^*)$, $\PoS_{sc}$ is at most $n$. Furthermore, there is a CACSG $(\Delta, \mathcal{F}^*)$ with $\PoS_{sc}=\PoSLsum$ even on parallel-link graphs.
\end{theorem}

 For the max-cost criterion, we show that $\PoS_{mc}$ is at most $n^2$.
 \begin{theorem}\label{thm:CACSG:PoS_mc}
For any CACSG $(\Delta, \mathcal{F}^*)$, $\PoS_{mc}$ is at most $n^2$. 
\end{theorem}

\begin{proof}
Let $\s^*$ be an optimal strategy profile under max-cost and $\s$ be a Nash equilibrium obtained from  $\s^*$ by the deviations of agents.
By the definition of the potential function, $\Phi(\s)\le \Phi(\s^*)$ holds.
By Equation (\ref{inequality:max:bound}) in Lemma \ref{lem:ubscpos}, $\Phi(\s^*)\le n\cdot\cost_{sc}(\s^*)$ holds.
Because ${\cost_{sc}(\s^*)}/{n} \le \cost_{mc}(\s^*)$, we have  $\Phi(\s)\le n^2\cdot \cost_{mc}(\s^*)$.
Finally, since $\cost_{mc}(\s)\le \Phi(\s)$, $\cost_{mc}(\s)\le n^2\cdot \cost_{mc}(\s^*)$ holds.
Thus, we have: $\PoS_{mc}\le {n^2\cdot\cost_{mc}(\s^*)}/{\cost_{mc}(\s^*)}=n^2$.
 \end{proof}

For the lower bound of $\PoS_{mc}$, Erlebach and Radoja showed that there is a CACSG $(\Delta, {F}_{\ord})$ with $\PoS_{mc}=\Omega(n\log n)$ \cite{Erlebach2015}.
However, there is a gap between $n\log n$ and $n^2$  with respect to $\PoS_{mc}$.



\section{Conclusion}\label{sec:conclusion}
In this paper, we studied the capacitated symmetric/asymmetric cost-sharing  connection game (CSCSG, CACSG) under the generalized cost-sharing scheme that models more realistic scenarios, such as where overhead costs arise.
In particular, we investigated the games from the viewpoint of $\PoA$ and $\PoS$ under two types of social costs: sum-cost and max-cost. 
Despite of the generalization, all the bounds under the ordinary cost-sharing function still hold with one exception; for $\PoS$ of sum-cost,  
we found a substantial difference between the ordinary cost-sharing function and the generalized scheme, where the former is $\log n$, and the latter is $\PoSLsum$. 

In this paper, we mainly focused on CSCSG. For the asymmetric games, however, there is still a gap for $\PoS$ under max-cost. Thus, filling the gap is an interesting open problem. Moreover, it would be worth considering some other concepts of stability, such as Strong Price of Anarchy (SPoA) and Strong Price of Stability (SPoS) to CSCSG and CACSG under the generalized cost-sharing scheme. 



\subsection*{Conflict of Interest}
The authors do not have any conflict of interest in hiring, financial support, or others.
\bibliographystyle{plainurl}
\bibliography{ref}

\end{document}